%% file: Mixed-permutation_channel4.tex
\documentclass[11pt]{article}
\input{ptm_preamble.tex}
\begin{document}

%==============================================================================================%
\title{\bf\large Mixed-permutation channel with its application to estimate quantum coherence}
%==============================================================================================%

\author{\blue{Lin Zhang}$^1$\footnote{E-mail: godyalin@163.com}\quad and\quad \blue{Ming-Jing Zhao}$^2$\footnote{E-mail: zhaomingjingde@126.com}\\
  {\it\small $^1$Institute of Mathematics, Hangzhou Dianzi University, Hangzhou 310018, PR~China}\\
  {\it\small $^2$School of
Science, Beijing Information Science and Technology University,
Beijing, 100192, PR~China}}
\date{}
\maketitle

\begin{abstract}
Quantum channel, as the information transmitter, is an indispensable
tool in quantum information theory. In this paper, we study a class
of special quantum channels named the \emph{mixed-permutation
channels}. The properties of these channels are characterized. The
mixed-permutation channels can be applied to give a lower bound of
quantum coherence with respect to any coherence measure. In
particular, the analytical lower bounds for $l_1$-norm
coherence and the relative entropy of coherence are shown respectively. The
extension to bipartite systems is presented for the actions of the
mixed-permutation channels.
\end{abstract}

\newpage

%================================%
\section{Introduction}
%================================%

Quantum channel is a deterministic quantum operation and it
transmits the information from the input to the output. It is
indispensable in various information processing tasks, such as
quantum communication \cite{Takagi2020}, quantum computation
\cite{Divincenzo1995,Ekert1996} and quantum cryptography
\cite{Gisin2002,Pirandola2020}. So the study about quantum channel
is an interesting and attractive topic. Explicitly, the resource
theory of quantum channels deals with understanding the properties
and capabilities of quantum channels in terms of the resources they
consume and produce
\cite{Chitambar2019,Wang2019,XWang2019,Kamin2020,Liu2020,Zhou2022}.
It also has important implications for the design of quantum
protocols, as it provides a unified framework \cite{Takagiphd} for
characterizing the capabilities and limitations of different types
of quantum channels. For quantum channels, the uncertainty
relations, a fundamental topic in quantum mechanics, are also
established recently. They describe the theoretical restrictions
about two or more quantum channels \cite{Fu,Zhou,Q.Zhang2021}, just
like that about two or more observables. In addition, quantum
channels are also intimately related to the dynamical behavior of
quantum states and the investigation of open quantum system
\cite{Bromley,Yu,Hu}.

One special kind of the most commonly-used quantum channels is the
projective measurement. Within the framework of the projective
measurement, many quantum features are manifested. For instance,
quantum coherence is such a quantum feature and it is of practical
significance in quantum computation and quantum communication
\cite{A.Streltsov-rev,E.Chitambar-review,M.Hu}. The quantum
coherence is further quantified by coherence measures such as  the
$l_1$-norm coherence \cite{Baumgratz2014}, the relative entropy of
coherence  \cite{Baumgratz2014}, the coherence concurrence
\cite{X.Qi}, the geometric coherence \cite{A.Streltsov2015} and the
robustness of coherence \cite{C.Napoli} and so on. However, the
evaluation of the coherence measure is not always easy for arbitrary
quantum states, especially for mixed states.

%Some commonly used channels are bit flip channel, phase flip
%channel, and depolarizing channel \cite{M.A.}.
In this paper, we investigate an interesting quantum channel called
the \emph{mixed-permutation channel}, which maps an input state to
an average of their permutation conjugations. This mixed-permutation
channel can be viewed as a special mixed-unitary channels with equal
weights \cite{J.Vicente}. More than that, it is also an example of a
covariant quantum channel \cite{Holevo2005}, meaning that it
commutes with certain symmetry operations on the input states. In
this case, this mixed-permutation channel is covariant with respect
to the conjugated action of the symmetric group $S_d$.

This work is factually a follow-up study of \cite{ZZ2022}, which
investigates the application of the mixed-permutation channel in
quantifying quantum coherence. Here we aim to explore the
mixed-permutation channel further and characterize it
systematically. Moreover, we find the application of the
mixed-permutation channel to evaluate the quantum coherence. This
evaluation is applicable to all coherence measures and is tight for
some special quantum states.

The paper is organized as follows. In Sect.~\ref{sect:2}, we give
the definition of the mixed-permutation channels and list their
basic properties. We also present the specific form of the
mixed-permutation channels, and characterize the image states. The
applications of the mixed-permutation channel in estimating quantum
coherence is also presented. Subsequently, in Sect.~\ref{sect:3}, we
extend the action of the mixed-permutation channels to bipartite
systems. We conclude this paper in Sect.~\ref{conclusion}.

%=========================================================================================%
\section{Mixed-permutation channel and its application in estimating quantum coherence}\label{sect:2}
%=========================================================================================%

In this section, we consider the $d$-dimensional quantum system with
its Hilbert space $\complex^d$ and fix the reference basis as
$\set{\ket{i}}_{i=1}^d$. The so-called quantum state is represented
by the density matrices/operators acting on the Hilbert space
$\complex^d$. We denote the set of all density matrices by
$$
\density{\complex^d}=\Set{\rho=\sum^d_{i,j=1} \rho_{ij} \out{i}{j}:
\rho\geqslant 0, \Tr{\rho}=1}.
$$
Any quantum channel $\Lambda$ acting on the density
matrices/operators in $\density{\complex^d}$ can be expressed with
Kraus representation $\Lambda(\cdot)=\sum_\mu K_\mu (\cdot)
K_\mu^\dagger$ where $\sum_\mu K_\mu^\dagger K_\mu=\I_d$, the
identity operator on $\complex^d$. Next we shall introduce the
mixed-permutation channels and characterize it both analytically and
geometrically. Then we explain its application in estimating quantum
coherence.

\subsection{Mixed-permutation channels}

\begin{definition}
For any matrix $\bsO$ acting on $\complex^d$, we define the
\emph{mixed-permutation channel} $\Delta$ as
\begin{eqnarray}
\Delta(\bsO):=\frac1{d!}\sum_{\pi\in
S_d}\bsP_\pi\bsO\bsP^\dagger_\pi,
\end{eqnarray}
where $S_d$ is the symmetric group of permutations of $d$ elements,
and the permutation matrix $\bsP_\pi$ is
\begin{eqnarray}
\bsP_\pi=\sum^d_{i=1}\out{\pi(i)}{i}.
\end{eqnarray}
induced by $\pi\in S_d$.
\end{definition}

By this definition we see for any matrix $\bsO$ acting on
$\complex^d$, the action of the mixed-permutation channel is the
average action over all possible permutation matrices. It is easily
seen that $\bsP_\pi\bsP_\sigma=\bsP_{\pi\sigma}$ and
$\bsP^\dagger_\pi=\bsP^\t_\pi=\bsP_{\pi^{-1}}$ for any $\pi$ and
$\sigma$ in $S_d$. Moreover, $\Delta(\bsO)=\frac1{d!}\sum_{\pi\in
S_d}\bsP^\dagger_\pi\bsO\bsP_\pi$. The properties of the
mixed-permutation channel $\Delta$ can be listed below:
\begin{enumerate}[(i)]
\item $\Delta \circ \Delta=\Delta$, that is $\Delta \circ \Delta(\bsO) =\Delta(\bsO)$ for any matrix $\bsO$ acting on $\complex^d$.
\item $\Delta$ is unital, that is $\Delta(\I_d)=\I_d$, where $\I_d$ is the identity matrix on $\complex^d$.
\item $\Delta$ is self-adjoint with respect to Hilbert-Schmidt inner product, i.e., $\Delta=\Delta^\dagger$ in the sense of $\Inner{\Delta^\dagger(\bsX)}{\bsY}=\Inner{\bsX}{\Delta(\bsY)}$ for all matrices $\bsX,\bsY$ acting on $\complex^d$, that is, $\Delta^\dagger(\bsO)=\frac1{d!}\sum_{\pi\in S_d}\bsP^\dagger_\pi\bsO\bsP_\pi=\Delta(\bsO)$.
\item $[\Delta(\bsO)]^\t = \Delta(\bsO^{\t})$, that is $[\Delta(\bsO)]^\t = (\frac1{d!}\sum_{\pi\in S_d}\bsP_\pi\bsO\bsP^\dagger_\pi)^\t=\frac1{d!}\sum_{\pi\in S_d}\bsP_\pi\bsO^\t\bsP^\dagger_\pi=\Delta(\bsO^{\t})$.
\item $\Delta$ is invariant under any permutation matrices, that is, $\bsP_\tau\Delta (\bsO)\bsP^\dagger_\tau=\Delta(\bsO)$ for any permutation matrix $\bsP_\tau$, induced by the permutation $\tau\in S_d$.
\end{enumerate}

Furthermore, the explicit form of the mixed-permutation channel on the
$\complex^d$ can be obtained analytically.

\begin{thrm}\label{th:1}
For any $d\times d$ matrix $\bsO=(o_{ij})$, the mixed-permutation
channel can be characterized as
\begin{eqnarray}
\Delta(\bsO)= \Tr{\bsO}\frac{\I_d}{d} +\Tr{\bsO(\bsE-\I_d)}
\frac{(\bsE-\I_d)}{d(d-1)},
\end{eqnarray}
where $\bsE:=\proj{\bse}$ for $d$-dimensional vector
$\bse=(1,1,\ldots,1)^\t$ with all entries being one.
\end{thrm}

\begin{proof}
In fact, this result will be established by linearity once we show
it holds for any Hermitian matrix $\bsO$ because every complex
square matrix $\bsM$ can be represented as a complex linear
combination of two Hermitian matrices, i.e.,
$\bsM=\bsH+\mathrm{i}\bsK$, where $\bsH:=\frac{\bsM+\bsM^\dagger}2$
and $\bsK=\frac{\bsM-\bsM^\dagger}{2\mathrm{i}}$. In what follows,
it suffices to consider Hermitian case.

Let us assume that $\bsO$ is
Hermitian. If $i\neq j$, then
\begin{eqnarray*}
\Innerm{i}{\Delta(\bsO)}{j} &=& \frac1{d!}\sum_{\pi\in
S_d}\Innerm{i}{\bsP^\dagger_\pi\bsO\bsP_\pi}{j} =
\frac1{d!}\sum_{\pi\in S_d}\Innerm{\pi(i)}{\bsO}{\pi(j)}=
\frac1{d!}\sum_{\pi\in S_d}o_{\pi(i)\pi(j)}.
\end{eqnarray*}
Now there exists a permutation $\tau_0\in S_d$ such that
$i=\tau_0(1)$ and $j=\tau_0(2)$ due to the assumption $i\neq j$.
Then
\begin{eqnarray*}
\sum_{\pi\in S_d}o_{\pi(i)\pi(j)} = \sum_{\pi\in
S_d}o_{\pi\tau_0(1)\pi\tau_0(2)}= \sum_{\pi\in S_d}o_{\pi(1)\pi(2)},
\end{eqnarray*}
implying that, whenever $i\neq j$, then
\begin{eqnarray*}
\Innerm{i}{\Delta(\bsO)}{j} = \Innerm{1}{\Delta(\bsO)}{2}.
\end{eqnarray*}
Similarly, if $i=j$, then
\begin{eqnarray*}
\Innerm{i}{\Delta(\bsO)}{i} = \frac1{d!}\sum_{\pi\in
S_d}o_{\pi(i)\pi(i)} =\frac1{d!}\sum_{\pi\in
S_d}o_{\pi(1)\pi(1)}=\Innerm{1}{\Delta(\bsO)}{1}.
\end{eqnarray*}
Note that $\Tr{\bsO}=\Tr{\Delta(\bsO)} =
\sum^d_{i=1}\Innerm{i}{\Delta(\bsO)}{i}=d\Innerm{1}{\Delta(\bsO)}{1}$,
which means that
\begin{eqnarray*}
\Innerm{i}{\Delta(\bsO)}{i} = \frac{\Tr{\bsO}}{d}.
\end{eqnarray*}
From the above discussion, let $\bsE:=\proj{\bse}$, where
$\bse=(1,\ldots,1)^\t$. We see that
\begin{eqnarray*}
\Delta(\bsO) = \frac{\Tr{\bsO}}d\I_d + \zeta(\bsE-\I_d).
\end{eqnarray*}
Then, via $\bsE^2=d\cdot\bsE$,
\begin{eqnarray*}
\bsE\Delta(\bsO)\bsE = \frac{\Tr{\bsO}}d\bsE^2 +
\zeta(\bsE^3-\bsE^2) = [\Tr{\bsO} + \zeta d(d-1)]\bsE.
\end{eqnarray*}
Because $\bsP_\pi\bsE=\bsE$, we get that $\bsE\Delta(\bsO)\bsE =
\bsE\bsO\bsE=\Innerm{\bse}{\bsO}{\bse}\bsE$. From these, we see that
\begin{eqnarray*}
\Innerm{\bse}{\bsO}{\bse} = \Tr{\bsO} + \zeta d(d-1),
\end{eqnarray*}
from which we get
\begin{eqnarray*}
\zeta =\frac{\Innerm{\bse}{\bsO}{\bse}-\Tr{\bsO}}{d(d-1)}
=\frac{2\re(o_{ij})}{d(d-1)}=\binom{d}{2}^{-1}\re(o_{ij}).
\end{eqnarray*}
Therefore
\begin{eqnarray*}
\Delta(\bsO) &=& \frac{\Tr{\bsO}}d\I_d + \zeta(\bsE-\I_d)\\
&=&  \frac{\Tr{\bsO}}d\I_d +
\frac{\Innerm{\bse}{\bsO}{\bse}-\Tr{\bsO}}{d(d-1)}(\bsE-\I_d)\\
&=& \frac{\Tr{\bsO}}d\I_d +
\frac{\Tr{\bsO\bsE}-\Tr{\bsO}}{d(d-1)}(\bsE-\I_d).
\end{eqnarray*}
After simplifying it then we get the desired expression. This
completes the proof.
\end{proof}

In Theorem~\ref{th:1}, when $\bsO$ is taken as a quantum state
$\rho\in \density{\complex^d}$, i.e., if a quantum state $\rho$ goes
through the mixed-permutation channel $\Delta$, then the output
state $\Delta(\rho)$ can be characterized analytically.
\begin{cor}\label{cor1}
For any quantum state $\rho=\sum^d_{i,j=1} \rho_{ij}
\out{i}{j}\in\density{\complex^d}$, the corresponding output state
$\rho_{\star}:=\Delta(\rho)$ from the mixed-permutation channel is
given by
\begin{eqnarray}\label{eqformrhop}
\rho_{\star}=\Pa{\begin{array}{cccccccc}
d^{-1} & \mathfrak{a} & \cdots & \mathfrak{a} \\
\mathfrak{a} & d^{-1} & \cdots & \mathfrak{a}\\
\vdots &\vdots & \ddots & \vdots \\
\mathfrak{a} & \mathfrak{a} &\cdots & d^{-1}
\end{array}},
\end{eqnarray}
where
\begin{eqnarray}\label{a}
\mathfrak{a}:= \frac{1}{d(d-1)} \sum_{i\neq j} \rho_{ij}
\end{eqnarray}
for the given reference basis $\set{\ket{i}:1,\ldots,d}$.
\end{cor}

\begin{remark}
We remark here that $\mathfrak{a}\in\Br{-\frac1{d(d-1)},\frac1d}$. Indeed,
\begin{eqnarray*}
\sum_{i\neq j} \rho_{ij} = \sum_{i,j=1}^d \rho_{ij}-1 =
\Tr{\rho\bsE}-1=\Innerm{\bse}{\rho}{\bse}-1.
\end{eqnarray*}
Therefore, combining
$d\lambda_{\min}(\rho)\leqslant\Innerm{\bse}{\rho}{\bse}\leqslant
d\lambda_{\max}(\rho)$, where $\lambda_{\min}(\rho)$ and $\lambda_{\max}(\rho)$ are the minimum and maximum eigenvalues of the quantum state $\rho$ respectively, we have,
\begin{eqnarray}\label{eq:parametera}
-\frac1{d(d-1)}\leqslant\frac{d\lambda_{\min}(\rho)-1}{d(d-1)}\leqslant
\mathfrak{a}=\frac1{d(d-1)}\sum_{i\neq j}\rho_{ij}\leqslant
\frac{d\lambda_{\max}(\rho)-1}{d(d-1)}\leqslant \frac1d.
\end{eqnarray}
\end{remark}

In order to see clearly what is the image states $\Delta(\rho)$ of
the channel $\Delta$ when input state is $\rho$, let us focus on
qubit systems. In fact, any qubit state can be represented as
\begin{eqnarray}\label{eq:Blochqubit}
\rho=\frac{1}{2}(\I_2+\bsr\cdot \boldsymbol{\sigma})=\frac{1}{2}
\left(
\begin{array}{cc}
1+r_3 & r_1-\mathrm{i}r_2\\
r_1+\mathrm{i}r_2 & 1-r_3
\end{array}
\right),
\end{eqnarray}
where $r_i\in \real$, $r:=\abs{\bsr}=\sqrt{r_1^2+r_2^2+r_3^2}
\leqslant 1$, and
$\bsr\cdot\boldsymbol{\sigma}:=\sum^3_{k=1}r_k\sigma_k$,
$\boldsymbol{\sigma}:=(\sigma_1,\sigma_2,\sigma_3)$ is the vector
of the Pauli matrices, given by
\begin{eqnarray}\label{eq:Pauli}
\sigma_1=\Pa{\begin{array}{cc}
               0 & 1 \\
               1 & 0
             \end{array}
},\quad\sigma_2=\Pa{\begin{array}{cc}
                      0 & -\mathrm{i} \\
                      \mathrm{i} & 0
                    \end{array}
},\quad \sigma_3=\Pa{\begin{array}{cc}
                       1 & 0 \\
                       0 & -1
                     \end{array}
}.
\end{eqnarray}
It is easily seen that the action of $\Delta$ is
\begin{eqnarray}\label{Drho2}
\Delta(\rho)=\frac{1}{2}(\I_2+\bsr^\prime\cdot \boldsymbol{\sigma})=\frac{1}{2}
\left(
\begin{array}{cc}
1 & r_1\\
r_1 & 1
\end{array}
\right).
\end{eqnarray}
So the action of the mixed-permutation channel on the set of quantum
states can be simulated by the transformation on the Bloch ball,
i.e.,
\begin{eqnarray}
\bsr=(r_1,r_2,r_3) \rightarrow \bsr^\prime=(r_1,0,0).
\end{eqnarray}
From Figure~\ref{fig:actiondelta}, we see that the Bloch ball is
mapped into a red line segment by the action of the
mixed-permutation channel.
\begin{figure}[h!]
\centering
\begin{minipage}[t]{0.4\textwidth}
\includegraphics[width=1\textwidth]{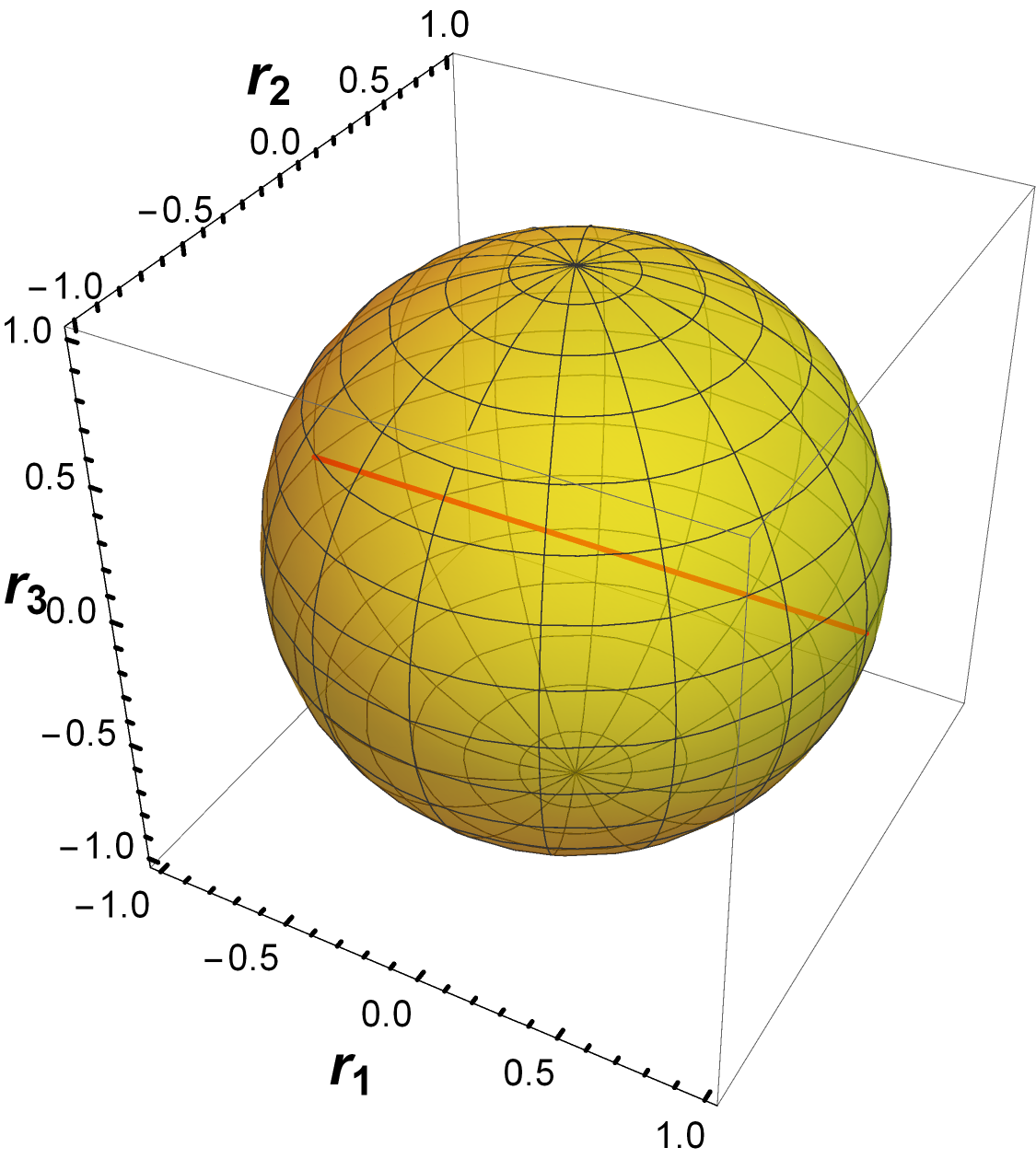}
\end{minipage}
\caption{(Color Online) The whole yellow Bloch ball stands for the
family of all qubit states in Eq. (\ref{eq:Blochqubit}); the red
line segment stands for the image states $\Delta(\rho)$ in Eq.
(\ref{Drho2}) when qubit states $\rho$ through the mixed-permutation
channel.}\label{fig:actiondelta}
\end{figure}

\subsection{Applying the mixed-permutation channel to estimate quantum coherence}

In order to estimate quantum coherence using the mixed-permutation
channel, we need recall some notions about quantum coherence. Given
the reference basis $\set{\ket{i}}_{i=1}^d$, a diagonal quantum
state $\varrho=\sum_{i = 1}^d \lambda_i\proj{i}$ is called the
incoherent state. The set of incoherent states is denoted as $\cI$.
The so-called incoherent operation is a completely positive linear
mapping $\Phi$ such that its Kraus decomposition
$\Phi(\rho)=\sum_{\nu} K_\nu\rho K_\nu^{\dagger}$ with $\sum_{\nu}
K_\nu^{\dagger} K_\nu=\I$ fulfills $K_\nu\delta
K_\nu^{\dagger}/\Tr{K_\nu\delta K_\nu^{\dagger}}\in \cI$ for all
$\delta\in \cI$ and for all $\nu$.

The quantum coherence is quantified by a nonnegative function named
the coherence measure. A coherence measure is required to satisfy
the following conditions \cite{Baumgratz2014}:
\begin{enumerate}
\item[(A1)] Nonnegativity: $C(\rho)\geqslant0$ for $\rho\in\density{\complex^d}$, and moreover $C(\rho)=0$ for all $\rho\in\cI$;
\item[(A2)] Monotonicity:
$C(\Phi(\rho))\leqslant C(\rho)$ for any incoherent operation
$\Phi$;
\item[(A3)] Strong monotonicity: $\sum_{\nu}p_{\nu}C(K_{\nu}\rho
K_{\nu}^{\dagger}/p_{\nu})\leqslant C(\rho)$ for any incoherent
operation $\Phi(\rho)=\sum_\nu K_\nu\rho K_\nu^{\dagger}$ with
$p_\nu=\Tr{K_{\nu}\rho K_{\nu}^{\dagger}}$ and
$\rho_{\nu}=K_{\nu}\rho K_{\nu}^{\dagger}/p_{\nu}$;
\item[(A4)] Convexity: $C(\rho)\leqslant \sum_{k} p_{k}C(\rho_{k})$ for
any quantum state $\rho=\sum_{k}p_{k}\rho_{k}$ with $ p_k\geqslant
0$ and $\sum_k p_k=1$.
\end{enumerate}
Two commonly used coherence measures are the $l_1$-norm coherence
and the relative entropy of coherence \cite{Baumgratz2014}. The
$l_1$-norm coherence of the quantum state $\rho=\sum_{i,j} \rho_{ij}
|i\rangle \langle j|$ is the sum of the magnitudes of all the
off-diagonal entries
\begin{eqnarray}
C_{l_1}(\rho):=\sum_{i\neq j} \abs{\rho_{ij}}.
\end{eqnarray}
The relative entropy of coherence is the difference of  von Neumann entropy between the density matrix and the diagonal matrix given by its diagonal entries,
\begin{eqnarray}
C_r(\rho):=S(\Pi(\rho))-S(\rho),
\end{eqnarray}
where $\Pi(\rho)=\diag(\rho_{11}, \rho_{22}, \cdots, \rho_{dd})$ is
the diagonal matrix obtained by the diagonal entries of $\rho$ and
$S(\rho):=-\Tr{\rho\ln\rho}$ is the von Neumann entropy of the state
$\rho$.

Now we come back to the mixed-permutation channel. It is obvious
that $\bsP_{\pi} \rho \bsP^\dagger_{\pi}$ is incoherent for any
incoherent state $\rho$. So the mixed-permutation channel is an
incoherent operation. By the monotonicity of the coherence measures,
we know the coherence of input state $\rho$ is nonincreasing under
this channel. Since the output states of the mixed-permutation
channel have an analytical form $\rho_{\star}$ in
Eq.~\eqref{eqformrhop}, we can utilize it to estimate the coherence
of the input state $\rho$.

\begin{thrm}\label{th:2}
For any coherence measure $C$ and any quantum state
$\rho\in\density{\complex^d}$, the quantum coherence of $\rho$ is
bounded from below by the quantum coherence of $\rho_{\star}$,
namely,
\begin{eqnarray}\label{eq:universal}
C(\rho)\geqslant C(\rho_{\star}).
\end{eqnarray}
\end{thrm}

\begin{proof}
First, since any permutation matrices  $\bsP_{\pi}$ and
$\bsP_{\pi^{-1}}$ are incoherent operations, so
$$
C(\rho)=C [\bsP^\dagger_{\pi}(\bsP_{\pi} \rho \bsP^\dagger_{\pi})
\bsP_{\pi}]\leqslant C (\bsP_{\pi} \rho \bsP^\dagger_{\pi})
\leqslant C(\rho).
$$
Therefore $C (\bsP_{\pi} \rho \bsP^\dagger_{\pi})=C(\rho)$ for any
permutation matrcix $\bsP_{\pi}$. Then by the convexity of coherence
measure, we have further that
\begin{eqnarray*}
C(\rho_{\star})=C(\Delta(\rho))&=&C\Pa{\frac1{d!}\sum_{\pi\in S_d}
\bsP_{\pi} \rho \bsP^\dagger_{\pi}}\\
&\leqslant& \frac1{d!}\sum_{\pi\in S_d} C \Pa{\bsP_{\pi} \rho
\bsP^\dagger_{\pi}} \\
&=& \frac1{d!}\sum_{\pi\in S_d} C (\rho)\\
&=&C(\rho),
\end{eqnarray*}
which completes the proof.
\end{proof}

From the proof above, we see that Theorem~\ref{th:2} is universal in
the sense that Eq.~\eqref{eq:universal} is independent of the
coherence measures. So Theorem~\ref{th:2} can be applied to all
coherence measures. Additionally, the output states $\rho_{\star}$
is of one parameter, so its coherence is much easy to calculate
compared with the input state.
In fact, for any quantum state $\rho=\sum^d_{i,j=1} \rho_{ij}
|i\rangle\langle j|$, the output state $\rho_{\star}$ in
Corollary~\ref{cor1} can be decomposed as
\begin{eqnarray}\label{mmcs}
\rho_{\star}=(1-\mathfrak{p})\I_d/d+ \mathfrak{p} \proj{\Phi_d},
\end{eqnarray}
where the weight is relevant to the quantum state as
\begin{eqnarray}\label{p}
\mathfrak{p}:=\frac1{d-1}\sum_{i\neq j}
\rho_{ij}\in\Br{-\frac1{d-1},1}
\end{eqnarray}
and $\ket{\Phi_d}:=\frac{1}{\sqrt{d}} \sum^d_{i=1} \ket{i}$ is the
{maximally coherent state} \cite{YPeng}. This state $\rho_{\star}$
is also called the \emph{maximally coherent mixed state} in
\cite{Singh2015}. By this decomposition in Eq.~\eqref{mmcs}, we get
the eigenvalues of $\rho_{\star}$ is
$\mathfrak{p}+\frac{1-\mathfrak{p}}d$ and $(1-\mathfrak{p})/d$ with
multiplicity $d-1$. In view of this, we get the following result:
\begin{cor}
For any quantum state $\rho\in\density{\complex^d}$, the
$l_1$-norm coherence and the relative entropy of coherence can be valuated from below by
\begin{eqnarray}\label{eq:universal}
C_{l_1}(\rho)&\geqslant&d(d-1)\abs{\mathfrak{a}}, \label{lower bound l1}\\
C_r(\rho)&\geqslant&
\Pa{1-\frac1d}(1-\mathfrak{p})\ln(1-\mathfrak{p})+\frac1d[(d-1)\mathfrak{p}+1]\ln[(d-1)\mathfrak{p}+1],
\end{eqnarray}
respectively, where $\mathfrak{a}$ and $\mathfrak{p}$ are depended
by $\rho$ as Eqs. \eqref{a} and \eqref{p}.
\end{cor}
For $l_1$ norm coherence, the lower bound in Eq.~\eqref{lower bound
l1} is tight for \emph{real} quantum states, that is,
\begin{eqnarray}
C_{l_1}(\rho)= C_{l_1}(\rho_{\star})
\end{eqnarray}
for any \emph{real} quantum state $\rho$. Now we consider the
$l_1$-norm coherence and the relative entropy of coherence in qubit
system.

\begin{exam}
In qubit systems, any qubit state $\rho$ can be represented as in
Eq.~\eqref{eq:Blochqubit} with the spectral decomposition as
$$
\rho=\lambda_1 \proj{\phi_1}+ \lambda_2 \proj{\phi_2}
$$
with the eigenvalues $\lambda_1=\frac{1}{2} (1+r)$,
$\lambda_2=\frac{1}{2} (1-r)$ and the corresponding eigenvectors
$$
\ket{\phi_1}=\frac{(r_1-\mathrm{i}r_2)\ket{0}-(r_3-r)\ket{1}}{\sqrt{2r(r-r_3)}},\quad
\ket{\phi_2}=\frac{(r_1-\mathrm{i}r_2)\ket{0}-(r_3+r)\ket{1}}{\sqrt{2r(r+r_3)}}.
$$
Let $r_2=r_3=\frac1{10}$, we calculate the $l_1$-norm coherence
$C_{l_1}(\rho)$ and the relative entropy coherence $C_{r}(\rho)$ as
well as the lower bound the $l_1$-norm coherence
$C_{l_1}(\rho_{\star})$ and the relative entropy coherence
$C_{r}(\rho_{\star})$ and make a comparison in Figures~\ref{fig:l1}
and \ref{fig:rel}.

\begin{figure}[htbp]
\centering
\subfigure[]%[The pure state decomposition $\mathfrak{D}^{(M)}$.]
{
\begin{minipage}[t]{0.45\textwidth}
\includegraphics[width=0.9\textwidth,height=0.7\textwidth]{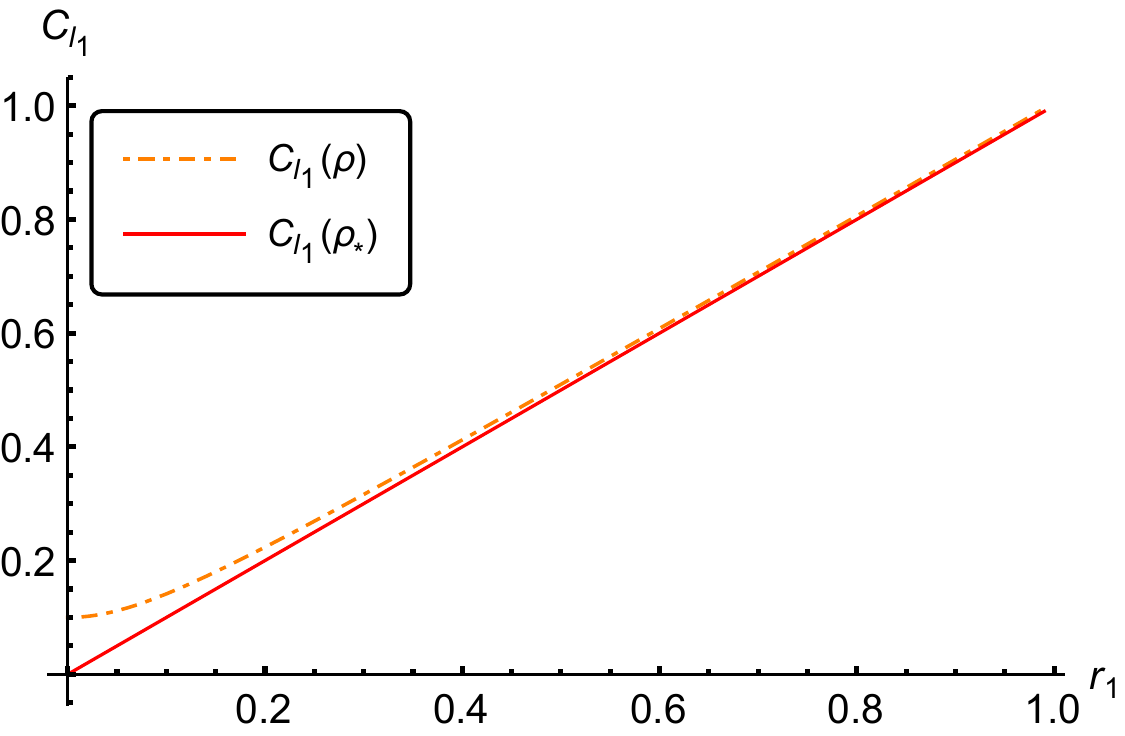}\label{fig:l1}
\end{minipage}
}
\quad
\subfigure[]%[The spectral decomposition $\mathfrak{D}^{(s)}$.]
{
\begin{minipage}[t]{0.45\textwidth}
\includegraphics[width=0.9\textwidth,height=0.7\textwidth]{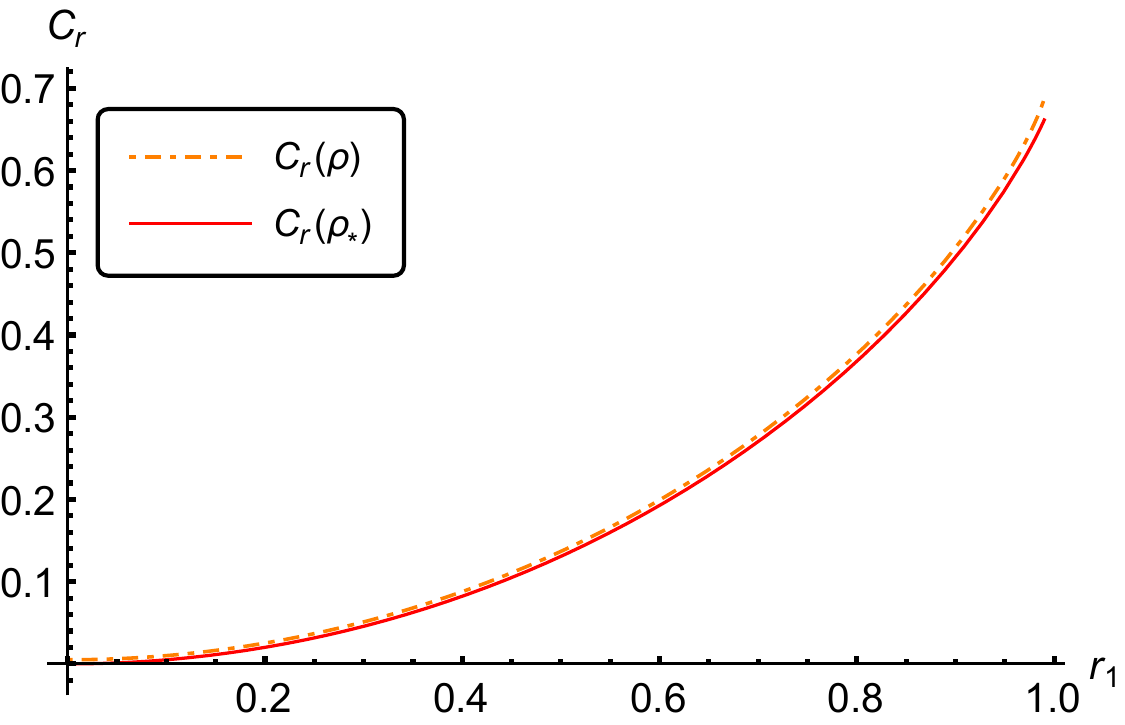}\label{fig:rel}
\end{minipage}
} \caption{(Color Online) The estimation of the coherence of $\rho$
by that of $\rho_{\star}$ using the $l_1$-norm coherence measure in
subfigure (a) and the relative entropy of coherence in subfigure
(b).}
\end{figure}
\end{exam}

From another point of view, it is
interesting to estimate quantum coherence from above.
Let us turn to another quantity called the \emph{coherence of
assistance}, induced by any coherence measure $C$, which is defined
in the form of concave bottom extension,
\begin{eqnarray}\label{Ca}
C_a(\rho)=\max \sum_k p_k C(\ket{\psi_k}),
\end{eqnarray}
where the maximization is taken over all pure state decompositions
of $\rho=\sum_k p_k \proj{\psi_k}$. Coherence of assistance
quantifies the coherence that can be extracted assisted by another
party under local measurements and classical communication
\cite{Chitambar2016}. Suppose Alice holds a state
$\rho^A=\rho=\sum_k p_k \proj{\psi_k}$ with coherence $C(\rho)$. Bob
holds another part of the purified state of $\rho$. The joint state
between Alice and Bob is $\sum_k p_k \ket{\psi_k}_A\ot\ket{k}_B$.
Bob performs local projective measurements $\Set{\proj{k}}$ along
the given basis and informs Alice the measurement outcomes by
classical communication. Alice's system will be in a pure state
ensemble $\Set{ p_k,\proj{\psi_k}}$ with average coherence $\sum_k
p_k C(\proj{\psi_k})$. The process is called assisted coherence
distillation. The maximum average coherence is called the coherence
of assistance which quantifies the one-way coherence distillation
rate \cite{Chitambar2016}.

%Since $\rho_{\star}=\Delta(\rho)$ can be viewed as the convex
%combination of maximally coherent state according to Eq.
%\eqref{mmcs} and therefore is the assisted maximally coherent state,
%which reaches the maximum of the coherence of assistance
%\cite{Zhao2021}.
The coherence of assistance of any quantum state
$\rho$ can be estimated by the output state
$\rho_{\star}=\Delta(\rho)$.

\begin{thrm}\label{th:3}
For any coherence measure $C$ and any quantum state
$\rho\in\density{\complex^d}$, the coherence of assistance of $\rho$
is bounded from above by the coherence of assistance of
$\rho_{\star}=\Delta(\rho)$,
\begin{eqnarray}
C_a(\rho)\leqslant C_a(\rho_{\star}).
\end{eqnarray}
%Especially, the equality holds true for two and three dimensional systems.
\end{thrm}

\begin{proof}
Note that
\begin{eqnarray*}
C_a(\rho_{\star})&=&C_a(\Delta(\rho))\\&=&C_a\Pa{\frac1{d!} \sum_{\pi\in
S_d} \bsP_{\pi} \rho \bsP^\dagger_{\pi}} \\
&\geqslant& \frac1{d!} \sum_{\pi\in S_d} C_a (\bsP_{\pi} \rho
\bsP^\dagger_{\pi})\\
&=&C_a(\rho).
\end{eqnarray*}
Here the inequality is the concavity of the coherence of assistance.
The last equality is due to the equality $C_a (\bsP_{\pi} \rho
\bsP^\dagger_{\pi})=C_a(\rho)$ for any permutation $\pi\in S_d$.
\end{proof}

%=====================================================================%
\section{Mixed-permutation channel on bipartite systems}\label{sect:3}
%=====================================================================%
In this section we study the mixed-permutation channel on bipartite
systems $\complex^{d_A} \ot \complex^{d_B}$. Without loss of
generality, we suppose $d_A\leqslant d_B$. Suppose that
$\set{\ket{i}}^{d_A}_{i=1}$ and $\set{\ket{j}}^{d_B}_{j=1}$ are
rothonormal bases of the subsystems $\complex^{d_A}$ and
$\complex^{d_B}$, respectively, we denote by
$$
\density{\complex^{d_A} \ot
\complex^{d_B}}=\Set{\rho_{AB}=\sum^{d_A}_{i,i'=1}\sum^{d_B}_{j,j'=1}\rho_{i'j',ij}\out{i'j'}{ij}:\rho_{AB}
\geqslant 0,\Tr{\rho_{AB}}=1}
$$
the set of density operators acting on $\complex^{d_A} \ot
\complex^{d_B}$. The identical channel on operator spaces are
denoted by $\mathrm{id}$. By Theorem~\ref{th:1}, we
 derive the explicit form of the mixed-permutation channel on bipartite systems $\complex^{d_A} \ot
\complex^{d_B}$.

\begin{thrm}\label{th:4}
For any bipartite operator $\bsX_{AB}$ acting on $\complex^{d_A} \ot
\complex^{d_B}$, denote $\bsE_{A(B)}=\proj{\bse_{A(B)}}$ where
$\bse_{A(B)}$ is a $d_{A(B)}$-dimensional column vector with all
entries being one, we have
\begin{eqnarray}
(\Delta_A\ot\mathrm{id}_B)
(\bsX_{AB})=\frac{\I_A}{d_A}\ot\Ptr{A}{\bsX_{AB}}+\frac{\bsE_A-\I_A}{d_A(d_A-1)}\ot\Ptr{A}{\bsX_{AB}(\bsE_A-\I_A)\ot\I_B}
\end{eqnarray}
and thus
\begin{eqnarray}
(\Delta_A\ot\Delta_B) (\bsX_{AB})&=&c_0(\bsX_{AB})\I_{AB} +c_1(\bsX_{AB})
\I_A\ot (\bsE_B-\I_B)+c_2(\bsX_{AB})
(\bsE_A-\I_A)\ot\I_B\notag\\
&&+c_3(\bsX_{AB})(\bsE_A-\I_A)\ot(\bsE_B-\I_B).
\end{eqnarray}
where
\begin{eqnarray}
\begin{cases}
c_0(\bsX_{AB})= \frac1{d_Ad_B}\Tr{\bsX_{AB}},\\
c_1(\bsX_{AB})= \frac1{d_Ad_B(d_B-1)}\Tr{\bsX_{AB}(\bsE_A-\I_A)\ot\I_B},\\
c_2(\bsX_{AB})= \frac1{d_Ad_B(d_A-1)}\Tr{\bsX_{AB}\I_A\ot(\bsE_B-\I_B)},\\
c_3(\bsX_{AB})=
\frac1{d_Ad_B(d_A-1)(d_B-1)}\Tr{\bsX_{AB}(\bsE_A-\I_A)\ot(\bsE_B-\I_B)}.
\end{cases}
\end{eqnarray}
\end{thrm}

\begin{proof}
The result can be also established by linearity once we show it
holds for any Hermitian matrix $\bsX_{AB}$. So it suffices to consider
the Hermitian case.
\begin{enumerate}[(i)]
\item In fact, for $\bsX_{AB}=\proj{\Psi_{AB}}$ where
$\ket{\Psi_{AB}}\in\complex^{d_A}\ot\complex^{d_B}$, we suppose its Schmidt
decomposition is
$$
\ket{\Psi_{AB}}=\sum_{i=1}^{d_A} \sqrt{\lambda_i}\ket{\bsa_i}_{A}\ot\ket{\bsb_i}_{B},
$$
with $\{\ket{\bsa_i}_{A}\}_{i=1}^{d_A}$ and $\{\ket{\bsb_j}_{B}\}_{j=1}^{d_B}$ the orthnormal bases of the subsystems $\complex^{d_A}$ and $\complex^{d_B}$ respectively,
then the density operator of $\proj{\Psi_{AB}}$ can be expressed as
$$
\proj{\Psi_{AB}}=\sum_{i,j=1}^{d_A}\sqrt{\lambda_i\lambda_j}\out{\bsa_i}{\bsa_j}_{A}\ot
\out{\bsb_i}{\bsb_j}_{B}.
$$
In the following reasoning, we omit the subindexes $A$ and $B$ when
no confusing arises. The action of the mixed-permutation channel on
the first subsystem is then
\begin{eqnarray*}
&&(\Delta_A\ot\mathrm{id}_B)
(\proj{\Psi_{AB}})\\&=&\sum_{i,j=1}^{d_A}\sqrt{\lambda_i\lambda_j}\Delta_A(\out{\bsa_i}{\bsa_j})\ot
\out{\bsb_i}{\bsb_j}\\
&=&\sum_{i,j=1}^{d_A}\sqrt{\lambda_i\lambda_j}
\Br{\Tr{\out{\bsa_i}{\bsa_j}}\frac{\I_{d_A}}{d_A}
+\Tr{\out{\bsa_i}{\bsa_j}(\bsE_A-\I_{d_A})}
\frac{(\bsE_A-\I_{d_A})}{d_A(d_A-1)}}\ot
\out{\bsb_i}{\bsb_j}\\
&=&\frac{\I_{d_A}}{d_A}\ot\Pa{
\sum_{i,j=1}^{d_A} \sqrt{\lambda_i\lambda_j}\Tr{\out{\bsa_i}{\bsa_j}}\out{\bsb_i}{\bsb_j}}\\
&&+\frac{(\bsE_A-\I_{d_A})}{d_A(d_A-1)}\ot\Pa{\sum_{i,j=1}^{d_A} \sqrt{\lambda_i\lambda_j}\Tr{\out{\bsa_i}{\bsa_j}(\bsE_A-\I_{d_A})}\out{\bsb_i}{\bsb_j}}\\
&=&\frac{\I_{d_A}}{d_A}\ot\ptr{A}{\proj{\Psi_{AB}}}+\frac{(\bsE_A-\I_{d_A})}{d_A(d_A-1)}\ot\ptr{A}{\proj{\Psi_{AB}}(\bsE_A-\I_{d_A})\ot\I_{d_B}}.
\end{eqnarray*}
By the linearity of both sides, this equality holds for any positive
semi-definite bipartite operator, and thus for a general bipartite
operator.
\item The proof goes similarly.
\end{enumerate}
We have done the proof.
\end{proof}

If we focus on the bipartite quantum states, then we get an output
state of the one-sided and two-sided mixed-permutation channels as
follows.

\begin{cor}\label{cor2}
For any quantum state $\rho_{AB}\in\density{\complex^{d_A} \ot
\complex^{d_B}}$, if the first subsystem goes through the mixed-permutation channel, then the output state is
\begin{eqnarray}\label{eq:univariate}
(\Delta_A\ot\mathrm{id}_B)
(\rho_{AB})=\frac{\I_{d_A}}{d_A}\ot\rho_B+\frac{d_A\proj{\Phi_{d_A}}-\I_{d_A}}{d_A(d_A-1)}\ot\Ptr{A}{\rho_{AB}(d_A
\proj{\Phi_{d_A}}-\I_{d_A})\ot\I_{d_B}},
\end{eqnarray}
where $\rho_B=\ptr{A}{\rho_{AB}}$. If both subsystems go through the
separate mixed-permutation channels, the output state is
\begin{eqnarray*}
(\Delta_A\ot\Delta_B)
(\rho_{AB})&=&\frac{\I_{d_A}}{d_A}\ot\frac{\I_{d_B}}{d_B}
+\gamma_1(\rho_{AB}) \frac{\I_{d_A}}{d_A}\ot
\frac{d_B\proj{\Phi_{d_B}}-\I_{d_B}}{d_B(d_B-1)}\\
&&+\gamma_2(\rho_{AB})\frac{d_A\proj{\Phi_{d_A}}-\I_{d_A}}{d_A(d_A-1)}\ot\frac{\I_B}{d_B}\\
&&+\gamma_3(\rho_{AB})\frac{d_A\proj{\Phi_{d_A}}
-\I_A}{d_A(d_A-1)}\ot\frac{d_B\proj{\Phi_{d_B}}-\I_{d_B}}{d_B(d_B-1)}.
\end{eqnarray*}
Here $\ket{\Phi_{d_A}}=\frac{1}{\sqrt{d_A}} \sum_{i=1}^{d_A}
\ket{\bsa_i}$ and $\ket{\Phi_{d_B}}=\frac{1}{\sqrt{d_B}} \sum_{j=1}^{d_B}
\ket{\bsb_j}$. Moreover,
\begin{eqnarray*}
\gamma_1(\rho_{AB})&:=&\Tr{\rho_{AB}(d_A\proj{\Phi_{d_A}}-\I_{d_A})\ot\I_{d_B}},\\
\gamma_2(\rho_{AB})&:=&\Tr{\rho_{AB}\I_{d_A}\ot(d_B\proj{\Phi_{d_B}}-\I_{d_B})},\\
\gamma_3(\rho_{AB})&:=&\Tr{\rho_{AB}(d_A\proj{\Phi_{d_A}}-\I_{d_A})\ot(d_B\proj{\Phi_{d_B}}-\I_{d_B})}.
\end{eqnarray*}
\end{cor}

By the form of the output states $(\Delta_A\ot\mathrm{id}_B)
(\rho_{AB})$ and $(\Delta_A\ot\Delta_B) (\rho_{AB})$ in the above
Corollary~\ref{cor2}, we get both $(\Delta_A\ot\mathrm{id}_B)
(\rho_{AB})$ and $(\Delta_A\ot\Delta_B) (\rho_{AB})$ are positive
under partial transpositions (PPT) for all states $\rho_{AB}$. Since
PPT states in qubit-qubit systems and qubit-qutrit systems are all
separable, so the action of the local mixed-permutation channel(s)
on two-qubit systems will completely erase the entanglement between
two subsystems.
\begin{exam}
For any quantum state $\rho_{AB}$ in $\density{\complex^{2}\ot\complex^{2}}$,
$$
\rho_{AB}=\Pa{\begin{array}{cccc}
             \rho_{11} & \rho_{12} & \rho_{13} & \rho_{14} \\
              \rho_{21} & \rho_{22} & \rho_{23} & \rho_{24} \\
               \rho_{31} & \rho_{32} & \rho_{33} & \rho_{34} \\
             \rho_{41} & \rho_{42} & \rho_{43} & \rho_{44}
           \end{array}
},
$$
%when we view it as a two-qubit state in
%$\density{\complex^{d_A}\ot\complex^{d_B}}$ for $d_A=d_B=2$,
if the
first qubit goes through the mixed-permutation channel, then the
output state is
\begin{eqnarray}
(\rho_{AB})_{\star,I} = (\Delta_A\ot\mathrm{id}_B)
(\rho_{AB})=\frac{\I_{d_A}}{d_A}\ot\rho_B+\frac{\bsE_A-\I_{d_A}}{d_A(d_A-1)}\ot\Ptr{A}{\rho_{AB}(\bsE_A-\I_{d_A})\ot\I_{d_B}},
\end{eqnarray}
with
$$
\rho_B=\ptr{A}{\rho_{AB}}=\Pa{\begin{array}{cc}
             \rho_{11}+\rho_{22} & \rho_{13}+\rho_{24} \\
             \rho_{31}+\rho_{42} & \rho_{33}+\rho_{44}
           \end{array}
}.
$$
More explicitly,
\begin{eqnarray}
(\rho_{AB})_{\star,I} = \Pa{
\begin{array}{cccc}
 \rho _{11}+\rho _{22} & \rho _{12}+\rho _{21} & \rho _{13}+\rho _{24} & \rho _{14}+\rho _{23} \\
 \rho _{12}+\rho _{21} & \rho _{11}+\rho _{22} & \rho _{14}+\rho _{23} & \rho _{13}+\rho _{24} \\
 \rho _{31}+\rho _{42} & \rho _{32}+\rho _{41} & \rho _{33}+\rho _{44} & \rho _{34}+\rho _{43} \\
 \rho _{32}+\rho _{41} & \rho _{31}+\rho _{42} & \rho _{34}+\rho _{43} & \rho _{33}+\rho _{44} \\
\end{array}
}
\end{eqnarray}
with $(\rho_{AB})^{\t_A}_{\star,I}=(\rho_{AB})_{\star,I}$.
%In the above
%calculation, we used the rule that $\bsA\ot \bsB=(b_{ij}\bsA)$.
Furthermore, if both two qubits  go through the mixed-permutation
channel, then the output state is
\begin{eqnarray*}
(\rho_{AB})_{\star,\star}&=&(\Delta_A\ot\Delta_B)(\rho_{AB})\\
&=&c_0(\rho)\I_{d_Ad_B}+c_1(\rho)\I_{d_A}\ot(\bsE_{B}-\I_{d_B})\\
&&+c_2(\rho)(\bsE_{A}-\I_{d_A})\ot\I_{d_B}+c_3(\rho)(\bsE_{A}-\I_{d_A})\ot(\bsE_{B}-\I_{d_B})\\
&=&\Pa{
\begin{array}{cccc}
 c_0(\rho_{AB}) & c_1(\rho_{AB}) & c_2(\rho_{AB}) & c_3(\rho_{AB}) \\
 c_1(\rho_{AB}) & c_0(\rho_{AB}) & c_3(\rho_{AB}) & c_2(\rho_{AB}) \\
 c_2(\rho_{AB}) & c_3(\rho_{AB}) & c_0(\rho_{AB}) & c_1(\rho_{AB}) \\
 c_3(\rho_{AB}) & c_2(\rho_{AB}) & c_1(\rho_{AB}) & c_0(\rho_{AB}) \\
\end{array}} = (\rho_{AB})^{\t_A}_{\star,\star}\geqslant0,
\end{eqnarray*}
where
\begin{eqnarray*}
c_0(\rho_{AB})&=&\frac14,\\%\quad
c_1(\rho_{AB})&=&\frac14\Tr{\rho_A(\bsE_{B}-\I_{d_B})},\\%\quad
c_2(\rho_{AB})&=&\frac14\Tr{\rho_B(\bsE_{A}-\I_{d_A})},\\
c_3(\rho_{AB})&=&\frac14\Tr{\rho_{AB}(\bsE_{A}-\I_{d_A})\ot(\bsE_{B}-\I_{d_B})},
\end{eqnarray*}
implying that $(\rho_{AB})_{\star,\star}$ must be separable. Moreover, the
eigenvalues are given by
\begin{eqnarray*}
c_0(\rho_{AB})+c_1(\rho_{AB})\pm(c_2(\rho_{AB})+c_3(\rho_{AB})),\quad
c_0(\rho_{AB})-c_1(\rho_{AB})\pm(c_2(\rho_{AB})-c_3(\rho_{AB})).
\end{eqnarray*}
From the above, we see that $(\rho_{AB})_{\star,I}$ and
$(\rho_{AB})_{\star,\star}$ are PPT states. Thus by Peres-Horodecki
criterion \cite{PPT1,PPT2}, the output states
$(\rho_{AB})_{\star,I}$ and $(\rho_{AB})_{\star,\star}$ are
separable ones for all two-qubit states $\rho_{AB}$.
\end{exam}

In what follows, let us recall the notion of the
entanglement-breaking channel \cite{Horodecki2003}. The so-called
entanglement-breaking channel $\cT\equiv\cT_A$ means that $(\cT_A\ot
\mathrm{id}_B)(\rho_{AB})$ is separable for every bipartite state
$\rho_{AB}\in\density{\complex^{d_A}\ot\complex^{d_B}}$. As already
known in \cite{Ruskai2003}, in qubit systems, any quantum channel is
entanglement breaking if and only if its Choi representation is
separable. Because $(\Delta_A\ot\mathrm{id}_B) (\rho_{AB})$ is
separable for any two-qubit state $\rho_{AB}$, a fortiori for
maximally entangled two-qubit state. Thus $\Delta$ is entanglement
breaking in qubit systems. In fact, the mixed-permutation channel
$\Delta$ is entanglement-breaking channel in qudit systems, which
can be summarized into the following result:
\begin{cor}
The mixed-permutation channel $\Delta$ characterized in
Theorem~\ref{th:1} is entanglement-breaking channel in qudit
systems.
\end{cor}

\begin{proof}
According to \cite[Theorem~4,
(B)$\Longleftrightarrow$(C)]{Horodecki2003}, it suffices to show
that the Choi-representation of the mixed-permutation channel $\Delta$,
$\jam{\Delta}:=(\Delta\ot\mathrm{id})(\proj{\Omega})$, is separable,
where $\ket{\Omega}:=\frac1{\sqrt{d}}\sum^d_{i=1}\ket{ii}$. Indeed,
using Eq.~\eqref{eq:univariate},
\begin{eqnarray*}
\jam{\Delta} &=& \frac1{d^2}\I_d\ot\I_d +
\frac1{d^2(d-1)}(d\proj{\Phi_d}-\I_d)^{\ot2}\\
&=& \frac1d\proj{\Phi_d}\otimes \proj{\Phi_d} +
\Pa{1-\frac1d}\Pa{\frac{\I_d-\proj{\Phi_d}}{d-1}} \otimes \Pa{\frac{\I_d-\proj{\Phi_d}}{d-1}},
\end{eqnarray*}
a separable state. Here $\ket{\Phi_d}=\frac{1}{\sqrt{d}}
\sum^d_{i=1} \ket{i}$. This indicates that $\Delta$ is
entanglement-breaking.
\end{proof}

\begin{exam}[The family of Bell-diagonal states]
The Bell-diagonal states \cite{Horodecki1996} in two-qubit system
can be written as
\begin{eqnarray}\label{rhobell}
\rho_{\text{Bell}}=\frac14(\I_4+\sum_{i=1}^3 t_i
\sigma_i\otimes\sigma_i)
\end{eqnarray}
with $\sigma_i$ three Pauli operators in Eq.~\eqref{eq:Pauli}. So a
Bell-diagonal state is specified by three real variables $t_1, t_2$,
and $t_3$ such that
\begin{eqnarray*}
\begin{cases}
1-t_1-t_2-t_3\geqslant 0,\\
1-t_1+t_2+t_3\geqslant 0,\\
1+t_1-t_2+t_3\geqslant 0,\\
1+t_1+t_2-t_3\geqslant 0.
\end{cases}
\end{eqnarray*}
Denote the set of all above such tuples $(t_1,t_2,t_3)$ by
$D_{\text{Bell}}$. Because the four eigenvalues of $\rho$ are in
$[0,1]$, we see that $t_i\in[-1,1](i=1,2,3)$. That is,
$D_{\text{Bell}}\subset[-1,1]^3$. The Bell-diagonal states can be
geometrically described by a tetrahedron. One can show that a
Bell-diagonal state is separable if and only if
$\abs{t_1}+\abs{t_2}+\abs{t_3}\leqslant 1$ holds. Geometrically, the
set of Bell-diagonal states is a tetrahedron and the set of
separable Bell-diagonal states is an octahedron
\cite{Horodecki1996}, which is denoted by $D_{\text{Bellsep}}$. That
is, $D_{\text{Bellsep}}=\set{(t_1,t_2,t_3)\in D_{\text{Bell}}:
\sum^3_{i=1}\abs{t_i}\leqslant 1}$. We find that
\begin{eqnarray}\label{eq:one-sided}
\Delta\ot\mathrm{id}(\rho_{\mathrm{Bell}}) =
\frac14(\I_4+t_1\sigma_1\ot\sigma_1).
\end{eqnarray}
We make a plot in arguments $(t_1,t_2,t_3)$ about the image states
in \eqref{eq:one-sided} of the one-sided action of the
mixed-permutation channel on the family \eqref{rhobell} of Bell
diagonal states in the following Figure~\ref{fig:action2delta}.
\begin{figure}[h!]
\centering
\begin{minipage}[t]{0.5\textwidth}
\includegraphics[width=1\textwidth]{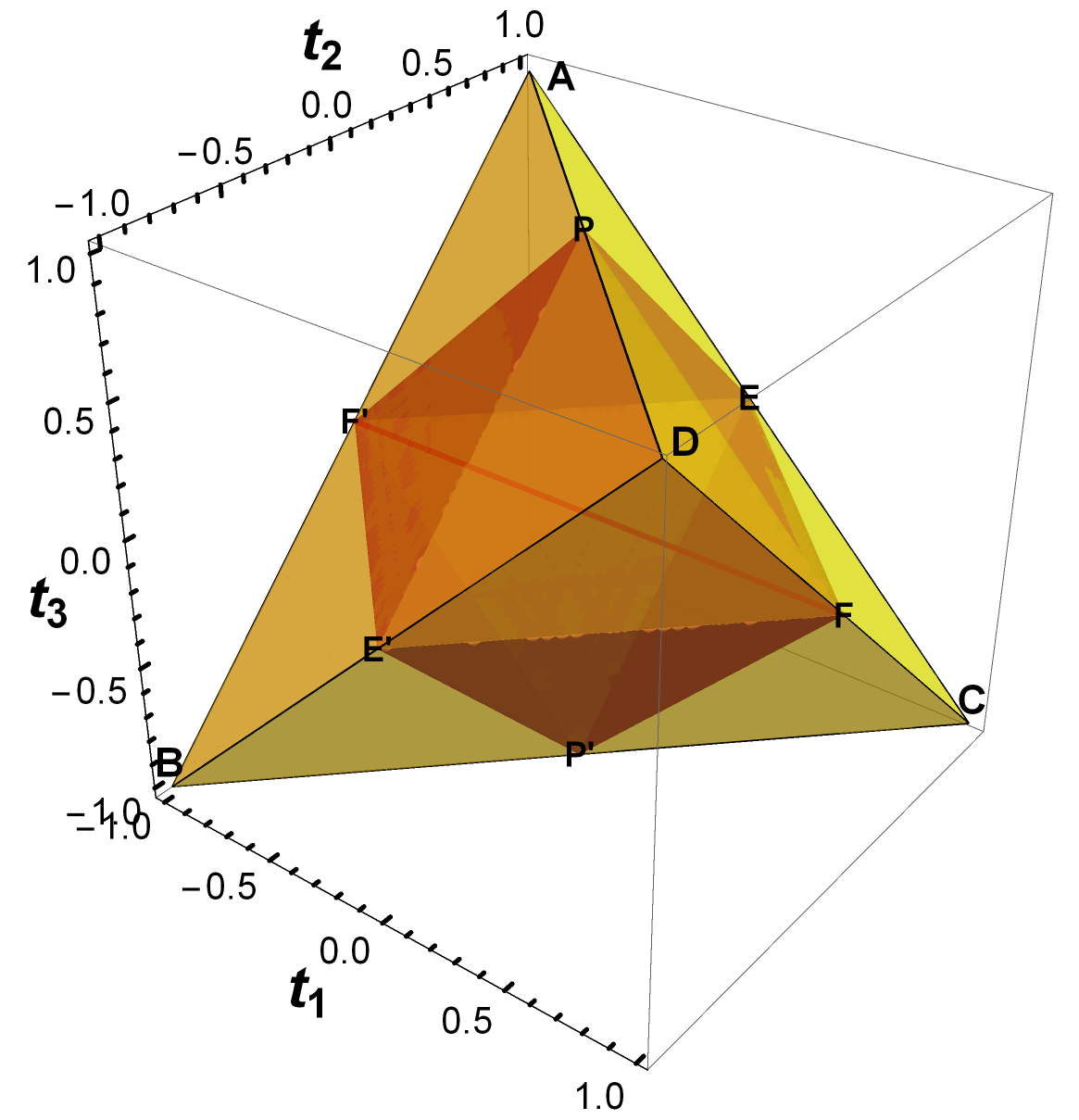}
\end{minipage}
\caption{(Color Online) The yellow tetrahedron stands for the family
of Bell-diagonal states in Eq. \eqref{rhobell}; the orange
octahedron stands for separable Bell-diagonal states; the red line
segment stands for image states
$\Delta\ot\mathrm{id}(\rho_{\mathrm{Bell}})$ in Eq.
\eqref{eq:one-sided} of all diagonal Bell-diagonal states through
the one-sided mixed-permutation channel.}
\end{figure}\label{fig:action2delta}
\end{exam}

%======================================================%
\section{Conclusions and discussions} \label{conclusion}
%======================================================%

In summary, we studied one kind of special channels, i.e., the
mixed-permutation  channels. The properties of this channel was
characterized. The mixed-permutation channel can be used in bounding
quantum coherence with respect to all coherence measures. The action
of the mixed-permutation channel on bipartite systems is also
discussed and can be generalized to multipartite systems. Beyond
that, there are some interesting problems which can be further
considered in the future:
\begin{itemize}
\item The first question is what is the average collective action of
the permutations on the bipartite quantum states,
\begin{eqnarray*}
\frac1{d!}\sum_{\pi\in S_d} (\bsP_\pi\ot\bsP_\pi)
\bsX(\bsP_\pi\ot\bsP_\pi)^\dagger=?.
\end{eqnarray*}
We are expecting a closed-form formula of the above expression. With
this potential formula, we can extend symmetrized variance for pure
states in \cite{ZZ2022} to mixed states.
\item The second question
is what is the correlation quantifiers \cite{N.Li} or coherence
measure \cite{Y.Sun} induced by the mixed-permutation channel.
\end{itemize}
In light of the special properties and the closed form of the output
state of the mixed-permutation channel, we expect the corresponding
correlation measure and coherence measure could characterize the
quantum features of quantum states from a permutation-invariant way.

%=============================================================================%
\subsubsection*{Acknowledgments}
This research is supported by Zhejiang Provincial Natural Science
Foundation of China under Grant No. LZ23A010005 and by NSFC under
Grant Nos.11971140 and 12171044.

%=============================================================================%
\subsubsection*{Data Availability Statement}
No Data associated in the manuscript.
%=============================================================================%

%=============================================================================%
\end{document}

%% file: ptm_preamble.tex
\newcommand{\jpa}{J. Phys. A~}

\newcommand{\njp}{New. J. Phys.~}

\newcommand{\prl}{Phys. Rev. Lett.~}
\newcommand{\pra}{Phys. Rev. A~}

\newcommand{\rmp}{Rev. Math. Phys.~}

\usepackage{epstopdf}
\usepackage{subfigure}
\usepackage[sc]{mathpazo}
\usepackage{amsmath}
\usepackage{amssymb}
\usepackage{graphicx,color}
\usepackage{multirow}
\usepackage{enumerate}
\usepackage{amsthm}% theorems and proofs
\usepackage{amsfonts,mathrsfs}
\usepackage{geometry} % allows easy specifying of the page layout
\usepackage{ifthen}
\usepackage{cite}
\usepackage[unicode=true,pdfusetitle, bookmarks=true,bookmarksnumbered=false,bookmarksopen=false, breaklinks=false,pdfborder={0 0 0},backref=false,colorlinks=false] {hyperref}
\hypersetup{
colorlinks,linkcolor=myurlcolor,citecolor=myurlcolor,urlcolor=myurlcolor}
\definecolor{myurlcolor}{rgb}{0,0,0.7}% Color highlighted for numbers in citations and theorems

\newcommand{\blue}{\textcolor{blue}}

\newcommand{\proj}[1]{| #1\rangle\!\langle #1 |}

\usepackage{hyperref}
\hypersetup{pdfpagemode=UseNone}

\newcommand{\tinyspace}{\mspace{1mu}}

\newcommand{\abs}[1]{\left\lvert\tinyspace #1 \tinyspace\right\rvert}

\renewcommand{\t}{{\scriptscriptstyle\mathsf{T}}}

\newcommand{\setft}[1]{\mathrm{#1}}

\newcommand{\density}[1]{\setft{D}\left(#1\right)}

\def \diag {\mathrm{diag}}

\def \re {\mathrm{Re}}

\def\complex{\mathbb{C}}
\def\real{\mathbb{R}}

\def\I{\mathbb{1}}

\def\ot{\otimes}

\newcommand{\out}[2]{| #1\rangle\langle #2 |}
\newcommand{\Inner}[2]{\left\langle #1 , #2\right\rangle}
\newcommand{\Innerm}[3]{\left\langle #1 \left| #2 \right| #3 \right\rangle}

\newcommand{\pa}[1]{(#1)}
\newcommand{\Pa}[1]{\left(#1\right)}

\newcommand{\Br}[1]{\left[#1\right]}
\newcommand{\set}[1]{\{#1\}}
\newcommand{\Set}[1]{\left\{#1\right\}}

\newcommand{\ket}[1]{|#1\rangle}

\def\Jamiolkowski{J}
\newcommand{\jam}[1]{\Jamiolkowski\pa{#1}}

\DeclareMathOperator{\trace}{Tr}
\newcommand{\ptr}[2]{\trace_{#1}\pa{#2}}
\newcommand{\Ptr}[2]{\trace_{#1}\Pa{#2}}

\newcommand{\Tr}[1]{\Ptr{}{#1}}

\def\cI{\mathcal{I}}

\def\cT{\mathcal{T}}

\def\bsE{\boldsymbol{E}}
\def\bsH{\boldsymbol{H}}
\def\bsK{\boldsymbol{K}}\def\bsM{\boldsymbol{M}}\def\bsO{\boldsymbol{O}}
\def\bsP{\boldsymbol{P}}
\def\bsX{\boldsymbol{X}}\def\bsY{\boldsymbol{Y}}

\def\bsa{\boldsymbol{a}}\def\bsb{\boldsymbol{b}}\def\bse{\boldsymbol{e}}

\def\bsr{\boldsymbol{r}}

\def\sE{\mathscr{E}}

\newtheorem{thrm}{Theorem}

\newtheorem{cor}{Corollary}
\theoremstyle{definition}
\newtheorem{definition}{Definition}
\newtheorem{remark}{Remark}
\newtheorem{exam}{Example}